\documentclass[11pt,reqno]{article}
\usepackage{latexsym, amsmath, amssymb, amsthm, a4, epsfig}
\usepackage{graphicx}
\usepackage{color}
\usepackage{url}

\newtheorem*{thm}{General Theorem}
\newtheorem{theorem}{Theorem}[section]
\newtheorem{cor}[theorem]{Corollary}
\newtheorem{lemma}{Lemma}
\newtheorem{prop}[theorem]{Proposition}
\newtheorem{definition}{Definition}[section]

\newtheorem{example}{Example}[section]
\newtheorem{remark}{Remark}[section]

\newcommand{\p}{\partial}

\newcommand{\Nbb}{\mathbb{N}}

\newcommand{\Rbb}{\mathbb{R}}

\newcommand{\Zbb}{\mathbb{Z}}

\newcommand{\beq}{\begin{equation}}
\newcommand{\eeq}{\end{equation}}

\numberwithin{equation}{section}
\numberwithin{figure}{section}

\begin{document}

\title{Circle Foliations Revisited: \\ Periods of Flows whose Orbits are all Closed}

\author{ Yoshihisa Miyanishi\thanks{Department of Mathematical Sciences,  Faculty of Science, Shinshu University,  Matsumoto 390-8621, Japan. Email: {\tt miyanishi@shinshu-u.ac.jp}.}}
\date{}
\maketitle

\begin{abstract}
We consider a $C^3$ free action $\mu: \Rbb\times M \to M$ on an $n$--dimensional $C^3$ compact connected manifold $M$.  
The manifold $M$ is foliated by circles if all orbits induced by $\mu$ are closed. The periods of $\mu$ are exactly   
the same if and only if there exists a Riemannian metric on $M$ with respect to 
which the circles are embedded as totally geodesic submanifolds of $M$ \cite{Wadsley}. 
It automatically follows a classical result \cite{Besse} that all periods are exactly the same if all the geodesic flows on the cosphere bundle  $S^* M$ are periodic (the so--called ``$P_T$--manifold'').  

Our purpose here is to adapt these results for Reeb flows or Hamiltonian flows on contact manifolds.    
Consequently, all periods are exactly the same if the contact manifold is connected and all orbits 
on the contact manifold are closed. 
We also present concrete examples of periodic flows, all of whose orbits are closed, such as Harmonic oscillators, 
Lotka-Volterra systems, and others. Lotka-Volterra systems, Reeb flows, and some geodesic flows have non-trivial periods, 
whereas the periods of Harmonic oscillators and similar systems can be easily obtained through direct calculations.

As an application to quantum mechanics, we examine the spectrum of semiclassical Shr\"odinger operators 
$$
P_{A, V} = \sum_{1\leq j \leq n}  (\hbar D_{x_j} - A_j({\bf{x}}))^2 +V({\bf {x}}), \quad {\bf x}=(x_1, \ldots, x_n) \in \Rbb^n.
$$
 If the spectrum of $P_{A, V}$ near $E \in \Rbb$ consists of eigenvalues, we denote such a set as 
$\Lambda_{c}(\hbar) = \{ E(\hbar)| \ |E-E(\hbar)| <c \hbar^{1-\delta},\ E(\hbar) \in \sigma_{\rm p}(P_{A, V})\}$ for small $\hbar>0$. For  fixed constants $c>0$ and $0<\delta<1/2$, the cluster of the spectrum near $E$ (the so--called essential\ difference\ spectrum) is defined by   
$$
\sigma_{\mbox{e.d.}} (E) := \lim_{\hbar\rightarrow +0} \overline{ \left\{ \frac{E_k(\hbar) -E_l(\hbar)}{\hbar}\Big| \ E_k(\hbar), E_l(\hbar) \in \Lambda_{c}(\hbar) \right\}}^{\rm ess},
$$
where the notation $\displaystyle\lim_{\hbar\rightarrow +0}\overline{\{\cdots\}}^{\rm ess}$ denotes the set of cluster (accumulation) points. 

Under suitable assumptions, we find that 
$$
\mbox{\rm either}\quad \sigma_{\mbox{e.d.}} (E)=\Rbb\quad \mbox{\rm or}\quad \displaystyle\sigma_{\mbox{e.d.}} (E)=\frac{2\pi}{T}\Nbb.
$$
In the former case, at least one of the Hamiltonian flows is {\it not periodic} on the energy surface $\Sigma_E=H^{-1}(E)$. 
In the latter case,  all Hamiltonian flows on $\Sigma_E=H^{-1}(E)$ are {\it periodic} with a common period $T$. This fact represents one of the semiclassical analogies of the Helton-Guillemin theorem \cite{Helton, Guillemin}. 
\end{abstract}



\section{Introduction}
Let $M$ be an $n$--dimensional $C^3$ manifold, and $\mu$ be a $C^3$ free action on $M$: 
\beq\label{eq: free action on M}
\begin{array}{ccc}
\mu: \Rbb \times M                    & \longrightarrow & M                     \\[-4pt]
\hspace{5mm}\rotatebox{90}{$\in$} & \hspace{-5mm}         & \rotatebox{90}{$\in$} \\[-4pt]
\hspace{5mm}  (t, x)                     &\hspace{-1.5mm} \longmapsto   \hspace{-1.5mm}  & \mu(t, x). 
\end{array}
\eeq
Then the closed compact orbits induced by the action $\mu$ are circles (referred to as $\mu$--circles). Our main concern is to find the nature of the periods of   
$\mu$ on a connected submanifold $N\subset M$ when $N$ is foliated by $\mu$--circles. 
In other words, we consider the period of $\mu$ if every flow on $N$ induced by $\mu$ is periodic. Notably,  
even if every flow is periodic, the minimal periods $T_{\mu}>0$ corresponding to the $\mu$--circles generally differ from each other. 
However, if $M$ is a $C^3$--Riemannian manifold and the free action generates the geodesic flow on $M$, 
the celebrated A. W. Wadsley's  theorem states that  all flows have the same period if the manifold $M$ is foliated by $\mu$--circles
 (see, e.g., \cite{Wadsley}, \cite{Epstein}, and \cite{Besse}; see also \cite{GM} for a recent development.). For the notations of Riemannian geometry, we refer the reader to \cite{KN}. 

Thus, if the action \eqref{eq: free action on M} restricted to $N$ can be considered as a geodesic flow of some Riemannian metric, we conclude  
that all flows on $N$ have the same period if $N$ is foliated by $\mu$--circles. 
To be more precise, we employ some terminolgy  from foliation theory \cite{Sullivan}. One term is that a one--dimensional foliation is {\it taut} 
if the leaves are geodesic for some Riemannian metric. Another term is that the flow, whose 
parametrization is arc-length, is called {\it geodesible} (see also Definition \ref{def: geodesible}).  This leads us to the following general theorem.

\begin{thm}\label{thm: main}
Let $N$ be a connected $C^3$ submanifold of $M$ and $\mu$ be a $C^3$ free action on $N$. Assume that $N$ is foliated by $\mu$--circles.   
If the circle foliation of $N$ is taut and geodesible, then all flows on $N$ induced by $\mu$ have the same period. 
\end{thm}
The above general theorem 
states a quintessential analogy to the geodesic flows on a $P_{T}$-manifold (see e.g. \cite{Besse}): 
All simple closed geodesics have the same length if all geodesics are closed (e.g., Sphere, Zoll surface).  

Our main purpose here is to adapt the general theorem 
for other flows frequently used in mathematical models, such as in Physics, Biology and others. To achieve this,
we explore systems and concrete examples that are both taut and geodesible. 

\begin{theorem}\label{thm: geodesible condition}
Let $N$ be a $(2n+1)$--dimensional connected $C^3$ manifold and  $\alpha$ be a contact one--form on $N$, 
that is, the form $\alpha \wedge (d \alpha)^n$ is nowhere vanishing on $N$.
Let $X$ be the Reeb field associated with $\alpha$, uniquely defined by $\iota_X d\alpha =0$ and $\alpha(X)=1$.    
If the orbits induced by the flow $\exp tX: N \rightarrow N$ are taut and closed, then all flows on $N$ induced by $\exp tX$ have the same period.  
\end{theorem}

We also refer the reader to \cite{BH} for the concepts of symplectic and contact manifolds.

Specifically, the general Hamiltonian $H$ on an $n$--dimensional manifold $M$ is considered:    
\beq\label{eq: Classical Hamiltonian}
 H(q, p) \in C^3(T^* M, \Rbb) 
 \eeq
where the notation $(q, p) \in T^* M$ denotes local coordinates in the cotangent space $T^*M$.  
It follows from Theorem \ref{thm: geodesible condition} that all Hamiltonian flows on the energy surface $\Sigma_E =\{(q, p) \in T^{*}M |\ H(q, p)=E\}$ has the same period if every Hamiltonian flow on $\Sigma_E$ is closed.
In fact, $\Sigma_E$ is a ($2n+1$)--dimensional contact manifold with some one--form $\alpha$ and 
the Reeb vector field is proportional to the Hamiltonian vector field $X_H$ (Proposition \ref{prop: geodesible condition Sullivan} and the following comments). We then notice that the energy surface $\Sigma_E$ can be decomposed as an {\it non-disjoint} union 
$\Sigma_E=\cup_{i} {\mathcal{F}}_i$. Here each circle foliation ${\mathcal{F}}_i$ induced by the Hamiltonian flow $\exp tX_H$ is taut 
(Proposition \ref{prop: almost geodesic foliation}).  Thus, all Hamiltonian flows on $\Sigma_E$ have the same period, if every flow is periodic on the energy surface $\Sigma_E$. 

\begin{theorem}\label{thm: Hamiltonian case}
Let $H \in C^3(T^*M, \Rbb)$ and $\Sigma_E := \{ (q, p)\in T^*M|\ H(q, p) =E\}$. Assume that $\Sigma_E$ is compact connected, $dH \not=0$ near $\Sigma_E$, and all Hamiltonian flows on $\Sigma_E$ are periodic. Then, all Hamiltonian flows on $\Sigma_E$ have the same period. 
\end{theorem}
We present concrete examples of periodic flows in section \ref{sec: example}. 

As an application to quantum mechanics, we consider Shr\"odinger Operators  on $L^2(\Rbb^n)$: 
\beq P_{A, V} = \sum_{1\leq j \leq n}  (\hbar D_{x_j} - A_j({\bf{x}}))^2 +V({\bf {x}}), \quad {\bf x}=(x_1, \ldots, x_n) \in \Rbb^n, \eeq
where $\hbar$ is Planck's constant, $D_{x_j} = \frac{1}{i} \frac{\partial}{\partial x_j}$, $A_j({\bf{x}})\ (j=1, \ldots, n)\in C^{\infty}(\Rbb^n)$, and $V({\bf{x}}) \in C^{\infty}(\Rbb^n).$  
Additional suitable assumptions are made in section \ref{sec: applications}. 
When the spectrum of $P_{A, V}$ near $E$ consists only of eigenvalues, we denote the set of such 
eigenvalues as $\Lambda_{c}(\hbar) = \{ E(\hbar)| \ |E-E(\hbar)| <c \hbar^{1-\delta},\ E(\hbar) \in \sigma_{\rm p}(P_{A, V})\}$ for small $\hbar>0$. For  fixed constants $c>0$ and $0<\delta<1/2$,  the cluster of the spectrum near $E$ (the so--called essential\ difference\ spectrum)  is defined by  
$$
\sigma_{\mbox{e.d.}} (E) := \lim_{\hbar\rightarrow +0} \overline{ \left\{ \frac{E_k(\hbar) -E_l(\hbar)}{\hbar} \Big| \ E_k(\hbar), E_l(\hbar) \in \Lambda_{c}(\hbar) \right\}}^{\rm ess}. 
$$
Here the notation $\displaystyle\lim_{\hbar\rightarrow +0}\overline{\{\cdots\}}^{\rm ess}$ denotes the set of cluster points as $\hbar \to +0$ (accumulation points), that is, 
$\tau \in \sigma_{\mbox{e.d.}} (E)$ if and only if for all small $\widetilde{\varepsilon} >0$ and small $\widetilde{\hbar}>0$, there exist $\hbar$--dependent eigenvalues $E_k(\hbar), E_l(\hbar) \in  \Lambda_{c}(\hbar)\ (\hbar<\widetilde{\hbar})$ such that $\left|\tau - \left| \frac{E_k(\hbar) -E_l(\hbar)}{\hbar} \right|\right|<\widetilde{\varepsilon}$. 

We then find that either $\sigma_{\mbox{e.d.}} (E)=\Rbb$ or $\displaystyle\sigma_{\mbox{e.d.}} (E)=\frac{2\pi}{T}\Nbb$.
In the former case, at least one of the Hamiltonian flows on $\Sigma_E$ is {\it not periodic}. 
In the latter case,  all Hamiltonian flows on $\Sigma_E$ are {\it periodic} with a common period $T$ (Theorem \ref{thm: semiclassical ev}). 

The remainder of this paper is organized as follows: Section \ref{sec: Proof of Theorems} describes the proofs of Theorem \ref{thm: geodesible condition} and Theorem \ref{thm: Hamiltonian case}. Concrete examples are presented in Section \ref{sec: example}. As an application, Section \ref{sec: applications} introduces a Helton-Guillemin type theorem for semiclassical Shr\"odinger operators on $L^2(\Rbb^n)$ (see \cite{Helton, Guillemin} for the relationship between geodesic flows and the essential difference spectrum of the Laplacian on a manifold). Some lemmas are presented in the appendices for the readers' convenience. 
\section{Proofs of Theorem \ref{thm: geodesible condition} and Theorem \ref{thm: Hamiltonian case} }\label{sec: Proof of Theorems}
As mentioned in the introduction, our next task is to identify the conditions under which the flow induced by $\mu$ is geodesible, and the corresponding $\mu$--foliation is taut. Let us recall a key definition and some essential propositions. 

\begin{definition}\label{def: geodesible}
Let $X$ be a nowhere-vanishing vector field on a manifold $M$. The vector field $X$ is said to be geodesible if there exists a Riemmanian metric $g$ such that
\beq
|X|=1 \quad \mbox{and} \quad \nabla_X X=0,
\eeq
where $\nabla$ denotes the Levi--Civita connection associated with $g$.  
\end{definition}

When no confusion arises, a one-dimensional foliation is also called geodesible (resp. geodesic) if it is spanned by a geodesible (resp. geodesic) vector field.  Here it is convenient to recall a proposition proposed by Sullivan \cite{Sullivan}. The detailed proof can be found in \cite{Becker, CV}. 

\begin{prop}\label{prop: geodesible condition Sullivan}
Let $X$ be a nowhere-vanishing vector field on a manifold $M$. Then $X$ is geodesible if and only if there esists a one--form $\alpha$ on $M$ 
such that $\alpha(X)=1$ and $i_X d\alpha =0$. 
\end{prop}

Thus, all Reeb fields are geodesible. 
For the readers' convenience, the proof of Proposition \ref{prop: geodesible condition Sullivan} is included in Appendix A. 
Theorem \ref{thm: geodesible condition} follows directly from Proposition \ref{prop: geodesible condition Sullivan} and the general theorem.  

To prove Theorem \ref{thm: Hamiltonian case}, we show that the energy surface $\Sigma_E$ is a contact manifold with a corresponding  contact form $\alpha$ such that the Reeb field $R$ is proportional to the Hamiltonian vector field $X_H$, that is, $R =(Const.) X_H$. To achieve this, we consider the cotangent bundle as a symplectic manifold $(T^*M, \omega)$. We also consider a local Liouville (or Euler) vector field $Y$ that is transversal to $\Sigma_E$, namely, $\mathcal{L}_Y \omega =\omega$ in a small  tubular neighborhood of $\Sigma_E$ (see e.g., \cite{BLM, C}). By defining $\alpha = i_Y \omega$, $d \alpha = d (\iota_Y \omega) =\omega$ and $\alpha \wedge (d\alpha)^{n-1} =  \iota_Y \omega \wedge \omega^{n-1} =\frac{1}{n} \iota_Y \omega^{n} $ is a volume form on any surface transverse to $Y$. 
It follows that $\alpha \wedge (d\alpha)^{n-1}$ is non-degenerate and $\alpha =\iota_Y \omega|_{\Sigma_E} $ is a contact form on $\Sigma_E$.  

We recall that the hypersurface $S$ is of contact type if and only if there exists 
a nowhere-vanishing function $f: S\rightarrow \Rbb$ such that $f X_H$ is a Reeb vector field. 
We find that $\Sigma_E$ is of contact type since we define $\alpha:=(\iota_Y \omega)|_{\Sigma_E}$ by a local Liouville field $Y$ and 
$\iota_{X_H} d\alpha := d\alpha|_{\Sigma_E} (X_H, \cdot)  = \omega|_{\Sigma_E} (X_H, \cdot) = dH|_{\Sigma_E} =0$. Thus the Hamiltonian vector field 
on $\Sigma_E$ becomes the Reeb field up to scale. 

Also we obtain from Cartan formula $\mathcal{L}_Y =d \iota_Y + \iota_Y d$ that 
\beq
\mathcal{L}_Y \omega = d \iota_Y \omega + \iota_Y d \omega =  d \iota_Y \omega +0 = \omega,
\eeq
and 
\beq\label{eq: Reeb flow defined by Liouville vector field}
\alpha(X_H) = \iota_Y \omega(X_H)=-\omega(X_H, Y)=-\iota_Y dH \not=0. 
\eeq

We define an alternative Liouville vector field by $\widetilde{Y}$ by $\widetilde{Y}=Y+X_K$ such that $\omega(\widetilde{Y}, X_H) =Const.$ where $X_K$ is a Hamiltonian vector field induced by a different Hamiltonian $K$. 
It then follows from \eqref{eq: Reeb flow defined by Liouville vector field} that $R=(Const. ) X_H$
under the corresponding contact form $\widetilde{\alpha} := (\iota_{\widetilde{Y}} \omega)|_{\Sigma_E}$.  
In fact, 
\beq
\mathcal{L}_{\widetilde{Y}} \omega = d \iota_{\widetilde{Y}} \omega + \iota_{\widetilde{Y}} d \omega =  d \iota_{\widetilde{Y}} 
\omega +0 =  d \iota_{{Y}} \omega + d \iota_{{X_K}} \omega =\omega+0 ,
\eeq
since $d(\iota_{X_K} \omega) = d(-dK)=0$. Thus $\widetilde{Y}$ is an alternative Liouville vevtor field. 

To solve $\omega(\widetilde{Y}, X_H) =Const.$, we choose a submanifold $\mathcal{F} \subset \Sigma_E$ that is 
foliated by circles induced by the Hamiltonian flow. We define the local Hamiltonian field $X_K$ on the neighbourhood of $\mathcal{F}$ such that $\omega(X_K, X_H) =Const.-\iota_Y dH$ near ${\mathcal{F}}$.  This is achieved using the method of characteristics near the submanifold $\mathcal{F}$, that is, $H$, $X_H$ and $\iota_Y dH$ are given and the first order PDE: $\omega(X_K, X_H) =Const. -\iota_Y dH $ is denoted as $dH (X_K) = Const. - \iota_Y dH$. 
This equation is solvable since  it follows from $dH \not=0$ near $\mathcal{F}$ that the $X_K$ is transversal to $\mathcal{F}$ for at least one suitable constant. Thus, we can find the Hamiltonian $K$ satisfying these conditions.  

As a result, the Hamiltonian vector fields and the Reeb fields are locally proportional to each other under a suitable one form $\widetilde{\alpha}$. 

\begin{prop}\label{prop: almost geodesic foliation}
Assume that $\Sigma_E$ is compact and foliated by $\mu$--circles. There exist at most finite many circle foliations $\{{\mathcal{F}}_s\}_s$ such that each ${\mathcal{F}}_s$ is taut and 
$\displaystyle \Sigma_E = \bigcup_{s} {\mathcal{F}}_s$. 
\end{prop}
We emphasize that $\Sigma_E\ (n\geq 2)$ cannot be regarded as a cosphere bundle on some manifold,  
even if the Hamiltonian system has a ``Kinetic + Potential" structure (see e.g. \cite{MMO}).  
We impose a metric structure only on each foliation $\mathcal{F}$ rather than on the entire cosphere bundle.
\begin{proof}[Proof of Proposition \ref{prop: almost geodesic foliation}]
It is known \cite{Sullivan} that a foliation is taut if and only if there exists a one form $\tau$ such that 
\beq\label{eq: foliation direction>0}
\tau\ \text{(each foliation direction)} >0 
\eeq 
and  
\beq\label{eq: tangent 2 plane field=0} 
d\tau\ \text{(any $2$--plane tangent to foliation)}=0.
\eeq 
For the proof of these facts, see Appendix B. 
To find such $\tau$ locally, let $\mathcal{F}_s$ be a small partial foliation in $N$ such that 
$S^1$ is a retraction of $\mathcal{F}_s$; that is, $\mathcal{F}_s \cong S^1 \times D_s$, where $D_s$ is contractible. 
As indicated by the de Rham cohomology ${\rm H^{1}_{dR}}(\mathcal{F}_s) \cong \Rbb$, there exists the so-called ``change in angle" type 
closed one--form  $\displaystyle d\theta =\tau$ satisfying the condition \eqref{eq: foliation direction>0}. 
Such a $\tau$ is closed and automatically satisfies condition \eqref{eq: tangent 2 plane field=0}. 

Let $\{D_s\}_s$ be a family of open sets such that $\displaystyle \bigcup_s \mathcal{F}_s \cong  \bigcup_s S^1 \times D_s$ covers $\Sigma_E$. Since $\Sigma_E$ is compact,  
at most a finite number of foliations of $\{\mathcal{F}_s\}_s$ cover $\Sigma_E$ as desired.
\end{proof}

\begin{prop}\label{prop: almost same period}
Under the same assumptions and notations as in Proposition \ref{prop: almost geodesic foliation}, if all flows on ${\mathcal{F}}_s$ are periodic, then all simple closed flows on each ${\mathcal{F}}_s$ have the same period.
\end{prop}

\begin{proof}[Proof of Proposition \ref{prop: almost same period}]
We find that each ${\mathcal{F}}_s$ is taut from Proposition \ref{prop: almost geodesic foliation}, namely, 
there exists some Riemannian metric on ${\mathcal{F}_s}$ with respect to which every leaf is embedded 
as a totally geodesic submanifold of ${\mathcal{F}}_s$.

As previously mentioned, the Hamiltonian vector field is locally proportional to the Reeb field for some contact form. 
It follows from Proposition \ref{prop: geodesible condition Sullivan}, Proposition \ref{prop: almost geodesic foliation}, and Theorem \ref{thm: geodesible condition} that Hamiltonian flows, which generate $C^r$ action $\mu: \Rbb \times {\mathcal{F}_s} \rightarrow \mathcal{F}_s$,   
correspond to a  $C^r$ action $\rho : S^1 \times  {\mathcal{F}_s}  \backslash S \rightarrow  {\mathcal{F}_s} $ with the same orbits.
\end{proof}

\begin{proof}[Proof of Theorem \ref{thm: Hamiltonian case}]
As established by Proposition \ref{prop: almost same period}, the period is constant on each $\mathcal{F}_s$.
The small foliations are connected, that is, each component $\mathcal{F}_s$ has another connected component $\mathcal{F}_{\tilde{s}}$ (i.e,. $\mathcal{F}_s \cap \mathcal{F}_{\tilde{s}} \not= \phi$ for some $\tilde{s} \not=s$). Therefore,  
all Hamiltonian flows on $\Sigma_E$ have the same period as desired.  
\end{proof}

As a consequence of Theorem \ref{thm: Hamiltonian case} and its proof, the periods of the Hamiltonian flows are locally constant:   
\begin{cor}
The period of the Hamiltonian flow on $\Sigma_E$ is constant on each connected foliated component. 
\end{cor}

\section{Examples}\label{sec: example}

Here, we present typical examples of periodic flows. 

\begin{example}[Harmonic Oscillator]
Let $M=\Rbb^n$ and $H(q, p) =\frac{p^2}{2m} + \frac{1}{2}k |q|^2 \in C^{\infty}(T^* \Rbb^n)$ $\quad (m>0, k>0)$.
Then the period is well--known to be $T=2\pi \sqrt{\frac{m}{k}}.$ This period $T$ is independent of the closed orbits.
\end{example}
In this case, one can find a transverse field of codimension-one planes that are invariant under the flow even on the phase space $T^* \Rbb^n \backslash O$ (See Fig. 1). 
Thus, all the flows except those at the origin have the same period, independent of $E>0$.  

\hspace{25mm}
{\unitlength 0.1in%
\begin{picture}(27.1000,28.1000)(26.2000,-43.3000)%
%
\special{pn 8}%
\special{pa 2620 3020}%
\special{pa 5330 3020}%
\special{fp}%
\special{sh 1}%
\special{pa 5330 3020}%
\special{pa 5263 3000}%
\special{pa 5277 3020}%
\special{pa 5263 3040}%
\special{pa 5330 3020}%
\special{fp}%
%
\special{pn 8}%
\special{pa 3910 4330}%
\special{pa 3910 1550}%
\special{fp}%
\special{sh 1}%
\special{pa 3910 1550}%
\special{pa 3890 1617}%
\special{pa 3910 1603}%
\special{pa 3930 1617}%
\special{pa 3910 1550}%
\special{fp}%
%
\special{pn 8}%
\special{ar 3910 3010 300 360 0.0000000 6.2831853}%
%
\special{pn 8}%
\special{ar 3910 3010 987 987 0.0000000 6.2831853}%
%
\special{pn 8}%
\special{ar 3920 3010 690 720 0.0000000 6.2831853}%
%
\special{pn 8}%
\special{pa 4060 2770}%
\special{pa 4060 2680}%
\special{fp}%
%
\special{pn 8}%
\special{pa 4060 2690}%
\special{pa 4160 2700}%
\special{fp}%
%
\special{pn 8}%
\special{pa 4400 2560}%
\special{pa 4400 2470}%
\special{fp}%
%
\special{pn 8}%
\special{pa 4400 2470}%
\special{pa 4510 2500}%
\special{fp}%
%
\special{pn 8}%
\special{pa 4650 2400}%
\special{pa 4650 2350}%
\special{fp}%
%
\special{pn 8}%
\special{pa 4650 2350}%
\special{pa 4730 2380}%
\special{fp}%
\put(51.4000,-29.3000){\makebox(0,0)[lb]{$q$}}%
\put(39.9000,-16.5000){\makebox(0,0)[lb]{$p$}}%
%
\special{pn 13}%
\special{pa 3780 2830}%
\special{pa 3080 2160}%
\special{fp}%
\special{pa 4130 2860}%
\special{pa 5110 2400}%
\special{fp}%
\put(43.1000,-18.8000){\makebox(0,0)[lb]{transversal plane invariant under the flow}}%
%
\special{pn 8}%
\special{pa 4970 1920}%
\special{pa 5010 2390}%
\special{fp}%
\special{sh 1}%
\special{pa 5010 2390}%
\special{pa 5024 2322}%
\special{pa 5005 2337}%
\special{pa 4984 2325}%
\special{pa 5010 2390}%
\special{fp}%
\special{pa 4240 1870}%
\special{pa 3190 2160}%
\special{fp}%
\special{sh 1}%
\special{pa 3190 2160}%
\special{pa 3260 2162}%
\special{pa 3241 2146}%
\special{pa 3249 2123}%
\special{pa 3190 2160}%
\special{fp}%
\put(38.0000,-31.6000){\makebox(0,0)[lb]{O}}%
\put(44.9000,-41.6000){\makebox(0,0)[lb]{Figure 1}}%
\end{picture}}%

\label{fig1}

Here, Bertrand's well-known theorem in classical mechanics states that the potentials $V(r)=-\frac{C_1}{r}$ or $V(r)={k}{r^2}$  have the property that all bound orbits are also closed orbits among central-force potentials (see e.g., \cite{OP, SST}). 
Although some such potentials are not smooth, we have the same period on the energy surface $\Sigma_E$.   
\begin{remark}[Kepler's law of planetary motion]
Let $M=\Rbb^3$ and $H(q, p) =\frac{p^2}{2m} - \frac{GMm}{|q|} \in C^{\infty}(T^* \Rbb^n \backslash \{q=0\})$, where $m, M>0, G>0$. 
It follows directly from Kepler's third law that $T^2 = \frac{4 \pi^2 a^3}{G(M+m)}$, where ${\displaystyle a}$ is the elliptical semi-major axis of the elliptical planetary orbit. Thus, all periods depend only on the total energy $E =-\frac{GMm}{2a}$ by direct calculations.

Notably, this fact can be demonstrated using the general theorem without direct calculations. 
Although Theorem \ref{thm: Hamiltonian case} is applicable only for $C^3$ Hamiltonians, the Kepler problem
in $\Rbb^3$ on an energy level E is equivalent to a reparameterization of the geodesic flow on the cosphere bundle 
of a manifold of constant curvature $-E$ with one point removed (see e.g., \cite{Milnor, Moser, MMO}). Thus,  
the periods depend only on the total energy $E$ from the general theorem.    

\end{remark}
As a non-trivial example of other fields, we introduce the property 
of the period of the Lotka-Volterra equations, which are known as the model of  
autocatalytic chemical reactions or the dynamics of biological systems: 
\begin{example}[Lotka-Volterra System]
Let us consider the system of differential equations 
\beq\label{eq: Lotka-Volterra eq.}
\dot{x_j} = \epsilon_j x_j + \frac{1}{\beta_j} \sum_{k=1}^{n} a_{jk} x_j x_k \quad (j=1, \ldots, n)
\eeq
as a model for the competition of $n$ biological species. Here $x_j =x_j(t)>0$ stands for the population of the $j$--th species, $A=(a_{jk})$ 
is a constant square matrix of order $N$, $\epsilon_j \in \Rbb$ are constants. Then, any solution ${\bf x}(t) =(x_j(t))$ to \eqref{eq: Lotka-Volterra eq.} 
remains in the positive cone $\Rbb^n_+ \equiv \{ (x_j) \in \Rbb^n | x_j>0\ (j=1, \ldots n)\}$, because each coordinate plane in $\Rbb^n$,  denoted as $x_j=0$ for some $j=1, \ldots, n$, is an invariant set of \eqref{eq: Lotka-Volterra eq.}.
We assume that there is a unique equilibrium ${\bf{q}}=(q_1, \ldots, q_{n}) \in \Rbb^n_+$ for the system 
satisfying 
\beq\label{eq: equibrium eq}
\epsilon_j + \sum_{k=1}^{n} a_{jk} q_k =0. \
\eeq
We introduce new variables $Q_j$ (referred to as the Volterra quantity of life) defined by the formula: 
\beq\label{eq: Volterra quantity of life}
Q_j=\int_{0}^t x_j(\tau)d\tau,
\eeq
and 
\beq\label{eq: Volterra moment}
P_j =\log \dot{Q_j} -\frac{1}{2}\sum_{k=1}^{n} a_{jk} Q_k,\quad (j=1, \ldots, n),
\eeq
(which are well-defined if the original system is restricted to $\Rbb^n_+$). The Hamiltonian for the system is then given by:  
\beq\label{eq: Lotka Volterra Hamiltonian}
H=\sum_{j=1}^{n}\epsilon_j Q_j  - \sum_{j=1}^n \exp(P_j +\frac{1}{2}\sum_{k=1}^n a_{jk} Q_k). 
\eeq
Computations show that the system \eqref{eq: Lotka-Volterra eq.} can be rewritten in Hamiltonian form as (See e.g., \cite{DFO, FO}.)
\beq\label{eq: Lotka Volterra Hamiltonian equations}
\begin{cases}
\dot{P_j} &=\frac{\partial H}{\p Q_j} \\
\dot{Q_j} &= -\frac{\partial H}{\p P_j} 
\end{cases}
\quad (j=1, \ldots, n).
\eeq
It follows from Proposition \ref{prop: geodesible condition Sullivan} and the associated comments the Hamilton flow induced by the corresponding vector fields $X_H$ is geodesible on $\Sigma_{E}=\{ ({\bf{Q}}, {\bf{P}}) \in T^{*}\Rbb_{+}^n |\ H({\bf{Q}}, {\bf{P}})=E \}$. If the  flow defined by \eqref{eq: Lotka-Volterra eq.} in $\Rbb^n_+$ is periodic, it is known \cite{DFO} that  
\beq 
\lim_{T\to\infty} \frac{1}{T} \int_0^{T} {\bf x}(t) dt = {\bf q}. 
\eeq
Thus ${\bf x}(t)$ is periodic if and only if $\widetilde{\bf{Q}}(t)={\bf{Q}}(t) - {\bf q} t \not= {\bf 0}$ is periodic. 
It follows from  \eqref{eq: Volterra moment}  that 
$\widetilde{\bf{P}}(t) = {\bf{P}}(t) + t (\frac{1}{2}A{\bf q})  = (P_j + t (\frac{1}{2} \sum_{k=1}^n a_{jk} q_k))$ is also periodic. 
The flow restricted to small foliations $({\bf{Q}}(t), {\bf{P}}(t)) \in \mathcal{F} \subset \Sigma_E$ is geodesible,  
and $ (\widetilde{\bf{Q}}(t), \widetilde{\bf{P}}(t)) = ({\bf{Q}}(t), {\bf{P}}(t))+ ({\bf{a}}, {\bf{b}})t$ is periodic with constant vectors 
${\bf{a}} , {\bf{b}} \in \Rbb^n$. 
\vspace{2mm}

\hspace{30mm}
{\unitlength 0.1in%
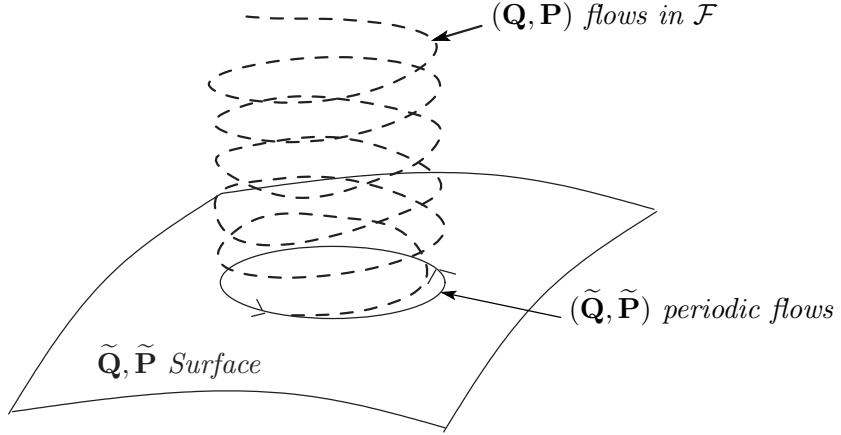
\begin{picture}(33.9100,22.2900)(14.2000,-27.1000)%
%
\special{pn 8}%
\special{pa 1420 2616}%
\special{pa 1437 2588}%
\special{pa 1454 2561}%
\special{pa 1488 2505}%
\special{pa 1505 2478}%
\special{pa 1523 2450}%
\special{pa 1540 2423}%
\special{pa 1557 2395}%
\special{pa 1593 2341}%
\special{pa 1610 2314}%
\special{pa 1646 2260}%
\special{pa 1665 2234}%
\special{pa 1683 2208}%
\special{pa 1740 2130}%
\special{pa 1760 2105}%
\special{pa 1779 2080}%
\special{pa 1799 2055}%
\special{pa 1819 2031}%
\special{pa 1882 1959}%
\special{pa 1926 1913}%
\special{pa 1948 1891}%
\special{pa 1994 1847}%
\special{pa 2017 1826}%
\special{pa 2041 1805}%
\special{pa 2066 1785}%
\special{pa 2090 1765}%
\special{pa 2115 1745}%
\special{pa 2140 1726}%
\special{pa 2166 1706}%
\special{pa 2192 1687}%
\special{pa 2218 1669}%
\special{pa 2244 1650}%
\special{pa 2270 1632}%
\special{pa 2324 1596}%
\special{pa 2350 1578}%
\special{pa 2377 1561}%
\special{pa 2405 1543}%
\special{pa 2432 1526}%
\special{pa 2459 1508}%
\special{pa 2486 1491}%
\special{pa 2514 1474}%
\special{pa 2534 1461}%
\special{fp}%
%
\special{pn 8}%
\special{pa 2534 1461}%
\special{pa 2566 1456}%
\special{pa 2598 1452}%
\special{pa 2631 1447}%
\special{pa 2663 1442}%
\special{pa 2695 1438}%
\special{pa 2727 1433}%
\special{pa 2760 1429}%
\special{pa 2792 1424}%
\special{pa 2824 1420}%
\special{pa 2856 1415}%
\special{pa 2889 1411}%
\special{pa 3017 1395}%
\special{pa 3050 1391}%
\special{pa 3082 1387}%
\special{pa 3146 1381}%
\special{pa 3178 1377}%
\special{pa 3274 1368}%
\special{pa 3306 1366}%
\special{pa 3338 1363}%
\special{pa 3434 1357}%
\special{pa 3466 1356}%
\special{pa 3498 1354}%
\special{pa 3562 1352}%
\special{pa 3594 1352}%
\special{pa 3626 1351}%
\special{pa 3690 1351}%
\special{pa 3721 1352}%
\special{pa 3753 1352}%
\special{pa 3817 1354}%
\special{pa 3848 1356}%
\special{pa 3880 1358}%
\special{pa 3911 1360}%
\special{pa 3943 1362}%
\special{pa 3975 1365}%
\special{pa 4006 1368}%
\special{pa 4038 1372}%
\special{pa 4100 1380}%
\special{pa 4132 1385}%
\special{pa 4163 1390}%
\special{pa 4195 1395}%
\special{pa 4288 1413}%
\special{pa 4320 1419}%
\special{pa 4382 1433}%
\special{pa 4413 1441}%
\special{pa 4444 1448}%
\special{pa 4568 1480}%
\special{pa 4599 1489}%
\special{pa 4631 1497}%
\special{pa 4662 1506}%
\special{pa 4693 1514}%
\special{pa 4755 1532}%
\special{pa 4786 1540}%
\special{pa 4798 1544}%
\special{fp}%
%
\special{pn 8}%
\special{pa 1440 2604}%
\special{pa 1472 2599}%
\special{pa 1505 2595}%
\special{pa 1569 2585}%
\special{pa 1601 2581}%
\special{pa 1633 2576}%
\special{pa 1666 2572}%
\special{pa 1698 2567}%
\special{pa 1762 2559}%
\special{pa 1795 2554}%
\special{pa 1923 2538}%
\special{pa 1956 2535}%
\special{pa 2020 2527}%
\special{pa 2180 2512}%
\special{pa 2212 2510}%
\special{pa 2244 2507}%
\special{pa 2340 2501}%
\special{pa 2372 2500}%
\special{pa 2404 2498}%
\special{pa 2500 2495}%
\special{pa 2627 2495}%
\special{pa 2723 2498}%
\special{pa 2754 2500}%
\special{pa 2786 2502}%
\special{pa 2817 2504}%
\special{pa 2849 2506}%
\special{pa 2881 2509}%
\special{pa 2912 2512}%
\special{pa 2944 2516}%
\special{pa 2975 2520}%
\special{pa 3007 2524}%
\special{pa 3069 2534}%
\special{pa 3101 2539}%
\special{pa 3194 2557}%
\special{pa 3226 2564}%
\special{pa 3257 2570}%
\special{pa 3288 2577}%
\special{pa 3319 2585}%
\special{pa 3350 2592}%
\special{pa 3443 2616}%
\special{pa 3475 2624}%
\special{pa 3506 2633}%
\special{pa 3537 2641}%
\special{pa 3568 2650}%
\special{pa 3599 2658}%
\special{pa 3692 2685}%
\special{pa 3704 2688}%
\special{fp}%
%
\special{pn 8}%
\special{pa 3697 2710}%
\special{pa 3714 2682}%
\special{pa 3731 2655}%
\special{pa 3765 2599}%
\special{pa 3783 2572}%
\special{pa 3800 2544}%
\special{pa 3817 2517}%
\special{pa 3835 2489}%
\special{pa 3852 2462}%
\special{pa 3906 2381}%
\special{pa 3924 2355}%
\special{pa 3942 2328}%
\special{pa 3961 2302}%
\special{pa 3979 2276}%
\special{pa 3998 2250}%
\special{pa 4018 2225}%
\special{pa 4037 2199}%
\special{pa 4057 2174}%
\special{pa 4077 2150}%
\special{pa 4097 2125}%
\special{pa 4118 2101}%
\special{pa 4138 2077}%
\special{pa 4160 2054}%
\special{pa 4181 2031}%
\special{pa 4203 2008}%
\special{pa 4248 1963}%
\special{pa 4272 1942}%
\special{pa 4295 1921}%
\special{pa 4343 1879}%
\special{pa 4368 1859}%
\special{pa 4393 1840}%
\special{pa 4418 1820}%
\special{pa 4443 1801}%
\special{pa 4495 1763}%
\special{pa 4521 1745}%
\special{pa 4548 1727}%
\special{pa 4574 1709}%
\special{pa 4655 1655}%
\special{pa 4682 1638}%
\special{pa 4710 1620}%
\special{pa 4764 1586}%
\special{pa 4792 1568}%
\special{pa 4811 1556}%
\special{fp}%
\put(18.8000,-24.3000){\makebox(0,0)[lb]{$\widetilde{\bf{Q}}, \widetilde{\bf{P}}$ Surface}}%
%
\special{pn 8}%
\special{ar 3116 1927 588 189 0.0000000 6.2831853}%
%
\special{pn 8}%
\special{pa 2747 2071}%
\special{pa 2719 2016}%
\special{fp}%
%
\special{pn 8}%
\special{pa 2761 2088}%
\special{pa 2692 2105}%
\special{fp}%
%
\special{pn 8}%
\special{pa 3615 1933}%
\special{pa 3656 1861}%
\special{fp}%
%
\special{pn 8}%
\special{pa 3677 1861}%
\special{pa 3759 1883}%
\special{fp}%
%
\special{pn 13}%
\special{pn 13}%
\special{pa 2903 2099}%
\special{pa 2966 2099}%
\special{fp}%
\special{pa 3031 2098}%
\special{pa 3086 2096}%
\special{pa 3094 2095}%
\special{fp}%
\special{pa 3158 2090}%
\special{pa 3180 2088}%
\special{pa 3212 2084}%
\special{pa 3221 2083}%
\special{fp}%
\special{pa 3285 2072}%
\special{pa 3310 2067}%
\special{pa 3344 2058}%
\special{pa 3346 2057}%
\special{fp}%
\special{pa 3408 2040}%
\special{pa 3413 2039}%
\special{pa 3449 2026}%
\special{pa 3468 2019}%
\special{fp}%
\special{pa 3526 1991}%
\special{pa 3544 1980}%
\special{pa 3569 1960}%
\special{pa 3576 1952}%
\special{fp}%
\special{pa 3607 1897}%
\special{pa 3610 1886}%
\special{pa 3610 1856}%
\special{pa 3606 1835}%
\special{fp}%
\special{pa 3582 1775}%
\special{pa 3575 1763}%
\special{pa 3555 1733}%
\special{pa 3546 1723}%
\special{fp}%
\special{pa 3500 1678}%
\special{pa 3480 1663}%
\special{pa 3453 1647}%
\special{pa 3447 1644}%
\special{fp}%
\special{pa 3387 1621}%
\special{pa 3363 1614}%
\special{pa 3332 1607}%
\special{pa 3325 1606}%
\special{fp}%
\special{pa 3261 1597}%
\special{pa 3234 1594}%
\special{pa 3200 1591}%
\special{pa 3199 1591}%
\special{fp}%
\special{pa 3134 1584}%
\special{pa 3134 1584}%
\special{pa 3100 1580}%
\special{pa 3072 1577}%
\special{fp}%
\special{pa 3007 1571}%
\special{pa 2967 1567}%
\special{pa 2944 1566}%
\special{fp}%
\special{pa 2880 1566}%
\special{pa 2870 1566}%
\special{pa 2838 1568}%
\special{pa 2817 1571}%
\special{fp}%
\special{pa 2754 1584}%
\special{pa 2744 1586}%
\special{pa 2714 1596}%
\special{pa 2694 1605}%
\special{fp}%
\special{pa 2637 1635}%
\special{pa 2626 1643}%
\special{pa 2598 1664}%
\special{pa 2587 1674}%
\special{fp}%
\special{pa 2543 1720}%
\special{pa 2531 1737}%
\special{pa 2519 1760}%
\special{pa 2516 1777}%
\special{fp}%
\special{pa 2538 1834}%
\special{pa 2539 1835}%
\special{pa 2558 1849}%
\special{pa 2583 1861}%
\special{pa 2593 1864}%
\special{fp}%
\special{pa 2655 1882}%
\special{pa 2681 1887}%
\special{pa 2717 1892}%
\special{fp}%
\special{pa 2782 1898}%
\special{pa 2809 1900}%
\special{pa 2845 1902}%
\special{fp}%
\special{pa 2909 1902}%
\special{pa 2953 1900}%
\special{pa 2973 1899}%
\special{fp}%
\special{pa 3037 1895}%
\special{pa 3053 1894}%
\special{pa 3100 1889}%
\special{fp}%
\special{pa 3164 1881}%
\special{pa 3199 1876}%
\special{pa 3226 1871}%
\special{fp}%
\special{pa 3290 1861}%
\special{pa 3293 1860}%
\special{pa 3337 1850}%
\special{pa 3352 1846}%
\special{fp}%
\special{pa 3414 1830}%
\special{pa 3422 1828}%
\special{pa 3461 1816}%
\special{pa 3474 1811}%
\special{fp}%
\special{pa 3535 1788}%
\special{pa 3565 1775}%
\special{pa 3592 1762}%
\special{fp}%
\special{pa 3647 1727}%
\special{pa 3663 1715}%
\special{pa 3678 1698}%
\special{pa 3691 1682}%
\special{fp}%
\special{pa 3698 1621}%
\special{pa 3696 1615}%
\special{pa 3686 1598}%
\special{pa 3672 1581}%
\special{pa 3660 1571}%
\special{fp}%
\special{pa 3608 1533}%
\special{pa 3606 1531}%
\special{pa 3578 1515}%
\special{pa 3553 1503}%
\special{fp}%
\special{pa 3494 1476}%
\special{pa 3481 1470}%
\special{pa 3446 1457}%
\special{pa 3435 1453}%
\special{fp}%
\special{pa 3374 1432}%
\special{pa 3373 1432}%
\special{pa 3336 1421}%
\special{pa 3313 1415}%
\special{fp}%
\special{pa 3250 1400}%
\special{pa 3231 1395}%
\special{pa 3198 1389}%
\special{pa 3188 1388}%
\special{fp}%
\special{pa 3124 1379}%
\special{pa 3104 1377}%
\special{pa 3061 1374}%
\special{fp}%
\special{pa 2997 1373}%
\special{pa 2985 1373}%
\special{pa 2954 1374}%
\special{pa 2934 1375}%
\special{fp}%
\special{pa 2869 1380}%
\special{pa 2859 1381}%
\special{pa 2826 1385}%
\special{pa 2807 1387}%
\special{fp}%
\special{pa 2742 1394}%
\special{pa 2716 1397}%
\special{pa 2680 1402}%
\special{fp}%
\special{pa 2616 1413}%
\special{pa 2602 1417}%
\special{pa 2568 1428}%
\special{pa 2557 1434}%
\special{fp}%
\special{pa 2510 1476}%
\special{pa 2504 1487}%
\special{pa 2499 1516}%
\special{pa 2501 1538}%
\special{fp}%
\special{pa 2515 1600}%
\special{pa 2516 1604}%
\special{pa 2528 1627}%
\special{pa 2543 1648}%
\special{pa 2548 1653}%
\special{fp}%
\special{pa 2599 1693}%
\special{pa 2604 1696}%
\special{pa 2628 1707}%
\special{pa 2656 1716}%
\special{pa 2658 1716}%
\special{fp}%
\special{pa 2721 1728}%
\special{pa 2749 1731}%
\special{pa 2783 1732}%
\special{pa 2784 1732}%
\special{fp}%
\special{pa 2849 1729}%
\special{pa 2857 1729}%
\special{pa 2895 1726}%
\special{pa 2912 1724}%
\special{fp}%
\special{pa 2975 1714}%
\special{pa 3016 1706}%
\special{pa 3037 1701}%
\special{fp}%
\special{pa 3100 1687}%
\special{pa 3143 1676}%
\special{pa 3161 1671}%
\special{fp}%
\special{pa 3224 1654}%
\special{pa 3229 1652}%
\special{pa 3271 1638}%
\special{pa 3284 1634}%
\special{fp}%
\special{pa 3345 1614}%
\special{pa 3356 1610}%
\special{pa 3398 1595}%
\special{pa 3405 1592}%
\special{fp}%
\special{pa 3465 1569}%
\special{pa 3479 1564}%
\special{pa 3518 1548}%
\special{pa 3523 1546}%
\special{fp}%
\special{pa 3582 1519}%
\special{pa 3588 1516}%
\special{pa 3618 1499}%
\special{pa 3637 1487}%
\special{fp}%
\special{pa 3680 1440}%
\special{pa 3684 1428}%
\special{pa 3683 1409}%
\special{pa 3675 1390}%
\special{pa 3669 1381}%
\special{fp}%
\special{pa 3623 1337}%
\special{pa 3617 1332}%
\special{pa 3589 1313}%
\special{pa 3570 1302}%
\special{fp}%
\special{pa 3513 1273}%
\special{pa 3485 1260}%
\special{pa 3455 1248}%
\special{fp}%
\special{pa 3394 1225}%
\special{pa 3367 1216}%
\special{pa 3334 1206}%
\special{fp}%
\special{pa 3271 1189}%
\special{pa 3254 1185}%
\special{pa 3219 1179}%
\special{pa 3209 1177}%
\special{fp}%
\special{pa 3146 1168}%
\special{pa 3120 1166}%
\special{pa 3082 1165}%
\special{fp}%
\special{pa 3018 1165}%
\special{pa 2999 1166}%
\special{pa 2955 1169}%
\special{fp}%
\special{pa 2891 1176}%
\special{pa 2855 1181}%
\special{pa 2828 1185}%
\special{fp}%
\special{pa 2764 1196}%
\special{pa 2763 1196}%
\special{pa 2730 1203}%
\special{pa 2702 1208}%
\special{fp}%
\special{pa 2639 1221}%
\special{pa 2621 1225}%
\special{pa 2584 1234}%
\special{pa 2578 1236}%
\special{fp}%
\special{pa 2519 1262}%
\special{pa 2508 1272}%
\special{pa 2505 1289}%
\special{pa 2513 1308}%
\special{pa 2520 1316}%
\special{fp}%
\special{pa 2568 1358}%
\special{pa 2590 1372}%
\special{pa 2622 1390}%
\special{fp}%
\special{pa 2680 1419}%
\special{pa 2706 1430}%
\special{pa 2739 1442}%
\special{fp}%
\special{pa 2800 1461}%
\special{pa 2821 1467}%
\special{pa 2857 1474}%
\special{pa 2862 1475}%
\special{fp}%
\special{pa 2926 1483}%
\special{pa 2927 1483}%
\special{pa 2960 1484}%
\special{pa 2989 1483}%
\special{fp}%
\special{pa 3053 1476}%
\special{pa 3056 1476}%
\special{pa 3086 1471}%
\special{pa 3116 1465}%
\special{fp}%
\special{pa 3178 1448}%
\special{pa 3205 1440}%
\special{pa 3234 1430}%
\special{pa 3238 1429}%
\special{fp}%
\special{pa 3299 1407}%
\special{pa 3320 1399}%
\special{pa 3349 1389}%
\special{pa 3358 1385}%
\special{fp}%
\special{pa 3419 1364}%
\special{pa 3436 1358}%
\special{pa 3466 1347}%
\special{pa 3478 1342}%
\special{fp}%
\special{pa 3537 1314}%
\special{pa 3552 1306}%
\special{pa 3579 1287}%
\special{pa 3589 1279}%
\special{fp}%
\special{pa 3636 1234}%
\special{pa 3655 1211}%
\special{pa 3672 1183}%
\special{fp}%
\special{pa 3680 1121}%
\special{pa 3669 1103}%
\special{pa 3645 1084}%
\special{pa 3635 1079}%
\special{fp}%
\special{pa 3575 1054}%
\special{pa 3542 1044}%
\special{pa 3515 1036}%
\special{fp}%
\special{pa 3452 1019}%
\special{pa 3438 1015}%
\special{pa 3405 1008}%
\special{pa 3391 1005}%
\special{fp}%
\special{pa 3327 992}%
\special{pa 3312 989}%
\special{pa 3265 981}%
\special{fp}%
\special{pa 3201 972}%
\special{pa 3195 971}%
\special{pa 3166 967}%
\special{pa 3138 964}%
\special{fp}%
\special{pa 3074 956}%
\special{pa 3024 953}%
\special{pa 3011 953}%
\special{fp}%
\special{pa 2947 952}%
\special{pa 2932 952}%
\special{pa 2899 954}%
\special{pa 2884 955}%
\special{fp}%
\special{pa 2819 963}%
\special{pa 2793 967}%
\special{pa 2757 974}%
\special{fp}%
\special{pa 2694 989}%
\special{pa 2671 995}%
\special{pa 2634 1007}%
\special{fp}%
\special{pa 2573 1028}%
\special{pa 2557 1035}%
\special{pa 2528 1050}%
\special{pa 2518 1058}%
\special{fp}%
\special{pa 2506 1110}%
\special{pa 2506 1111}%
\special{pa 2527 1126}%
\special{pa 2556 1139}%
\special{pa 2561 1141}%
\special{fp}%
\special{pa 2622 1159}%
\special{pa 2632 1162}%
\special{pa 2674 1171}%
\special{pa 2684 1173}%
\special{fp}%
\special{pa 2748 1182}%
\special{pa 2762 1184}%
\special{pa 2806 1188}%
\special{pa 2811 1188}%
\special{fp}%
\special{pa 2875 1192}%
\special{pa 2895 1192}%
\special{pa 2938 1191}%
\special{fp}%
\special{pa 3003 1187}%
\special{pa 3066 1182}%
\special{fp}%
\special{pa 3130 1173}%
\special{pa 3157 1168}%
\special{pa 3192 1161}%
\special{fp}%
\special{pa 3255 1146}%
\special{pa 3279 1140}%
\special{pa 3316 1130}%
\special{fp}%
\special{pa 3377 1110}%
\special{pa 3390 1106}%
\special{pa 3425 1093}%
\special{pa 3437 1088}%
\special{fp}%
\special{pa 3496 1062}%
\special{pa 3517 1052}%
\special{pa 3544 1037}%
\special{pa 3551 1033}%
\special{fp}%
\special{pa 3605 997}%
\special{pa 3612 992}%
\special{pa 3630 977}%
\special{pa 3645 962}%
\special{pa 3652 954}%
\special{fp}%
\special{pa 3678 897}%
\special{pa 3678 891}%
\special{pa 3675 878}%
\special{pa 3669 865}%
\special{pa 3660 853}%
\special{pa 3649 844}%
\special{fp}%
\special{pa 3593 811}%
\special{pa 3593 811}%
\special{pa 3570 802}%
\special{pa 3544 794}%
\special{pa 3533 791}%
\special{fp}%
\special{pa 3471 776}%
\special{pa 3455 772}%
\special{pa 3421 767}%
\special{pa 3408 765}%
\special{fp}%
\special{pa 3344 757}%
\special{pa 3311 754}%
\special{pa 3281 752}%
\special{fp}%
\special{pa 3217 748}%
\special{pa 3190 747}%
\special{pa 3154 747}%
\special{fp}%
\special{pa 3089 747}%
\special{pa 3062 748}%
\special{pa 3026 750}%
\special{fp}%
\special{pa 2962 753}%
\special{pa 2899 759}%
\special{fp}%
\special{pa 2834 766}%
\special{pa 2800 771}%
\special{pa 2772 776}%
\special{fp}%
\special{pa 2708 787}%
\special{pa 2676 794}%
\special{pa 2647 800}%
\special{fp}%
\special{pa 2584 817}%
\special{pa 2571 821}%
\special{pa 2543 831}%
\special{pa 2525 839}%
\special{fp}%
\special{pa 2473 875}%
\special{pa 2468 884}%
\special{pa 2469 895}%
\special{pa 2476 905}%
\special{pa 2487 915}%
\special{pa 2500 923}%
\special{fp}%
\special{pa 2560 947}%
\special{pa 2574 951}%
\special{pa 2605 959}%
\special{pa 2621 962}%
\special{fp}%
\special{pa 2685 972}%
\special{pa 2716 976}%
\special{pa 2747 979}%
\special{fp}%
\special{pa 2812 983}%
\special{pa 2846 985}%
\special{pa 2875 985}%
\special{fp}%
\special{pa 2940 984}%
\special{pa 2940 984}%
\special{pa 2988 982}%
\special{pa 3003 981}%
\special{fp}%
\special{pa 3067 975}%
\special{pa 3085 973}%
\special{pa 3130 966}%
\special{fp}%
\special{pa 3193 954}%
\special{pa 3228 947}%
\special{pa 3255 940}%
\special{fp}%
\special{pa 3317 923}%
\special{pa 3319 922}%
\special{pa 3362 908}%
\special{pa 3377 902}%
\special{fp}%
\special{pa 3437 879}%
\special{pa 3443 876}%
\special{pa 3480 858}%
\special{pa 3494 851}%
\special{fp}%
\special{pa 3550 820}%
\special{pa 3574 804}%
\special{pa 3598 786}%
\special{pa 3601 783}%
\special{fp}%
\special{pa 3645 736}%
\special{pa 3649 731}%
\special{pa 3657 714}%
\special{pa 3661 697}%
\special{pa 3659 681}%
\special{pa 3657 677}%
\special{fp}%
\special{pa 3612 633}%
\special{pa 3602 627}%
\special{pa 3577 616}%
\special{pa 3554 608}%
\special{fp}%
\special{pa 3491 591}%
\special{pa 3482 589}%
\special{pa 3446 581}%
\special{pa 3430 578}%
\special{fp}%
\special{pa 3366 568}%
\special{pa 3330 564}%
\special{pa 3303 561}%
\special{fp}%
\special{pa 3239 555}%
\special{pa 3212 553}%
\special{pa 3176 550}%
\special{fp}%
\special{pa 3111 547}%
\special{pa 3103 547}%
\special{pa 3069 546}%
\special{pa 3048 546}%
\special{fp}%
\special{pa 2983 545}%
\special{pa 2973 545}%
\special{pa 2920 545}%
\special{fp}%
\special{pa 2856 544}%
\special{pa 2851 544}%
\special{pa 2821 544}%
\special{pa 2792 543}%
\special{fp}%
\special{pa 2728 539}%
\special{pa 2696 537}%
\special{pa 2665 534}%
\special{fp}%
%
\special{pn 13}%
\special{pa 3895 583}%
\special{pa 3649 656}%
\special{fp}%
\special{sh 1}%
\special{pa 3649 656}%
\special{pa 3719 656}%
\special{pa 3700 641}%
\special{pa 3707 618}%
\special{pa 3649 656}%
\special{fp}%
\put(39.4300,-6.1100){\makebox(0,0)[lb]{$(\bf{Q}, \bf{P})$ flows in $\mathcal{F}$}}%
%
\special{pn 8}%
\special{pa 4310 2110}%
\special{pa 3690 1980}%
\special{fp}%
\special{sh 1}%
\special{pa 3690 1980}%
\special{pa 3751 2013}%
\special{pa 3742 1991}%
\special{pa 3759 1974}%
\special{pa 3690 1980}%
\special{fp}%
\put(43.5000,-21.5000){\makebox(0,0)[lb]{$(\widetilde{\bf{Q}}, \widetilde{\bf{P}})$  periodic flows}}%
\end{picture}}%

It follows from equation \eqref{eq: Volterra quantity of life} that each elements of ${\mathbf{Q}} \in \Rbb_{+}^n$ is monotone increasing since $x_j(t) >0\ \mbox{for}\ t>0$. We impose a Riemannian metric on $\Sigma_E$ such that each flow is geodesible. 
There exists a small foliation $\mathcal{F} \subset \Sigma_E$ that can be considered a Riemannian covering of the small submanifold $N$ where $(\widetilde{\mathbf{Q}}, \widetilde{\mathbf{P}})$ flows exist.  We define the Riemannian covering 
as $\pi: {\mathcal{F}} \rightarrow N$, that is, the following conditions hold: (1) The map $\pi$ is a smooth covering map. 
(2) The map $\pi$ induces a local isometry. This ensures that the local foliation, which consists of the oribits 
$(\widetilde{\bf{Q}}(t), \widetilde{\bf{P}}(t))$, is geodesible.

The taut condition is the same as in the proof of Proposition \ref{prop: almost geodesic foliation}. Thus the periods of the Lotka--Volterra equation \eqref{eq: Lotka-Volterra eq.} depend only on the 
energy of Volterra Hamiltonian \eqref{eq: Lotka Volterra Hamiltonian}. 
(See e.g. \cite{KSY} for the sufficient conditions under which all solutions of \eqref{eq: Lotka-Volterra eq.} are periodic). 
\end{example}

\section{An application to semiclassical Schr\"odinger eigenvalues.}\label{sec: applications}
As an application of Therem \ref{thm: Hamiltonian case}, we can find the so-called difference spectrum of eigenvalues for semiclassical operators (See e.g.  \cite{Kuwabara} for the essential difference spectrum). We refer readers to \cite{DS} and the reference therein for the notion of semiclassical analysis.  
Although we may consider more general $h$--admissible operators, for simlicity, we focus only on magnetic Schr\"odinger operators with smooth potentials, usually denoted as  
\beq\label{eq: semiclassical Shroedinger op.}
P_{A, V}(\hbar) = \sum_{1\leq j \leq n}  (\hbar D_{x_j} - A_j({\bf{x}}))^2 +V({\bf {x}}), \quad {\bf x}=(x_1, \ldots, x_n) \in \Rbb^n. 
\eeq
where the number $\hbar \in \Rbb_{>0}$ represents Planck's constant. 
We make the following assumptions: 

(A1). \ The corresponding principal symbol $H({\bf{x}}, {\bf{p}}) =\sum\limits_{1\leq j \leq n}  (p_j - A_j({\bf{x}}))^2 +V({\bf {x}})  \in C^{\infty}(T^*\Rbb^n)$ is bounded from below. Consequently, $P_{A, V}(\hbar)$ is essentially self--adjoint in $L^2(\Rbb^n)$, starting from 
$C_0^{\infty}(\Rbb^n).$

(A2).\ $E$ is not a critical value of $H({\bf{x}}, {\bf{p}})$, that is, $dH \not=0$ near $\Sigma_{E}=H^{-1}(E)$.


(A3).\ The energy surface $\Sigma_{E}=H^{-1}(E)$ is connected.

(A4).\ There exists $\varepsilon >0$ such that $H^{-1} ([E-\varepsilon, E+\varepsilon])$ is compact in $T^{*}\Rbb^n$.

(A5).\ $\exists \gamma>0, \exists C>0$ and $\exists M>0$ such that $|V(x)| \leq C (V(y)+\gamma)(1+|x-y|)^{M}$ for all $x, y \in \Rbb^n$.

(A6).\ $\exists \gamma>0$ such that $\forall \alpha$, $\exists c_{\alpha}$, $|\partial^\alpha V(x)| \leq c_{\alpha}|V(x) +\gamma|$ and $|\partial^\alpha {\bf A}({x})|  \leq c_{\alpha}|V(x) +\gamma|^{1/2}$ for all $x \in \Rbb^n$.

Then the spectrum of $P_{A, V}$ near $E$ consists only of eigenvalues. We denote the set of such 
eigenvalues as $\Lambda_{c}(\hbar) = \{ E(\hbar)| \ |E-E(\hbar)| <c \hbar^{1-\delta},\ E(\hbar) \in \sigma_{\rm p}(P_{A, V})\}$ for small $\hbar>0$. For  fixed constants $c>0$ and $0<\delta<1/2$,  the cluster of the spectrum near $E$ (the so--called essential\ difference\ spectrum)  is defined by  
$$
\sigma_{\mbox{e.d.}} (E) := \lim_{\hbar\rightarrow +0} \overline{ \left\{ \frac{E_k(\hbar) -E_l(\hbar)}{\hbar} \Big| \ E_k(\hbar), E_l(\hbar) \in \Lambda_{c}(\hbar) \right\}}^{\rm ess}. 
$$
where the notation $\displaystyle\lim_{\hbar \rightarrow +0}\overline{\{\cdots\}}^{\rm ess}$ denotes the set of cluster points (accumulation points). Then we find the following theorem.

\begin{theorem}[Guillemin-Helton Type theorem for semiclassical Schr\"odinger eigenvalues]\label{thm: semiclassical ev}
Under assumptions (A1) to (A6), 
\beq
\mbox{\rm either}\ \sigma_{\mbox{\rm e.d.}}(E)=\Rbb\quad \mbox{\rm or}\quad {\sigma_{\mbox{\rm e.d.}} (E)} = \frac{2\pi }{T} \Zbb.  
\eeq
In the former case, at least one of Hamiltonian flows is not periodic. In the latter case,  all Hamiltonian flows induced by $H$ are periodic on $\Sigma_{E}=H^{-1}(E)$ and 
the number $T>0$ denotes the minimal period. 
\end{theorem}

\begin{proof}[Outline of Proof]

The spectrum of the Schr\"odinger operators \eqref{eq: semiclassical Shroedinger op.} has already been thoroughly investigated for smoother readability. In fact, it is known \cite{CR} that $$ \mbox{\rm either}\ \sigma_{\mbox{\rm e.d.}}(E)=\Rbb\quad \mbox{\rm or}\quad \mbox{all  Hamiltonian flows are periodic on}\ \Sigma_E .$$
This result is analogous to Helton's theorem \cite{Helton} on compact manifolds. We also refer to \cite{Miyanishi} and references therein.

If all Hamiltonian flows are periodic on $\Sigma_E$, Theorem \ref{thm: Hamiltonian case} implies that all periods are exactly  
the same. Consequently, according to the celebrated result \cite{Dozias} that 
for fixed constants $b >0$ and $0<\delta<1/2$, there exists constant $C>0$ such that 
\beq\label{eq: Shrodinger spectrum near E} 
\mbox{Spec}(P_{A, V}(\hbar)) \cap [E-b\hbar^{1-\delta}, E+b\hbar^{1-\delta}] \subset \bigcup_{k\in \Zbb} I_k^{\delta}(\hbar) \quad \forall\ \mbox{small}\ \hbar >0,   
\eeq
where 
\beq\label{eq: Shroedinger spectrum band}
I_k^{\delta}(\hbar)  = \left[ -\frac{S}{T} +\frac{\hbar}{T}\left( -\frac{\pi}{2}\mu + 2\pi k \right) -C h^{2-2\delta}.
-\frac{S}{T} +\frac{\hbar}{T}\left( -\frac{\pi}{2}\mu + 2\pi k\right) -C h^{2-2\delta} \right]. 
\eeq
Here $S$ (resp. $\mu$) denotes the action (resp. the Maslov index) of a periodic trajectory. We notice that 
the action $S$ and the Maslov index $\mu$ are independent of trajectories on $\Sigma_{E}$ since $\Sigma_E$ is connected.

The formulas \eqref{eq: Shrodinger spectrum near E} and \eqref{eq: Shroedinger spectrum band} indicate that the eigenvalues near $E$ is denoted as 
\beq\label{eq: eigenvalue asymptotics}
E_k(\hbar) \approx 
-\frac{S}{T} +\frac{\hbar}{T}\left( -\frac{\pi}{2}\mu + 2\pi k \right)\quad k\in \Zbb\quad \mbox{modulo}\ \hbar^{2-2\tilde{\delta}}. 
\eeq
for some $0 < \tilde{\delta} < 1/2$, we emphasize that the result of B. Helffer and D. Robert on the distribution \cite{HD1} 
assures the existence of $E_k(\hbar).$  In fact, S. Dozias employs their results to derive the formulas 
\eqref{eq: Shrodinger spectrum near E} and \eqref{eq: Shroedinger spectrum band}.  It then follows from \eqref{eq: eigenvalue asymptotics} that the  differences of eigenvalues are denoted as 
$\frac{E_k(\hbar)-E_l (\hbar)}{\hbar} \approx \frac{2\pi (k-l)}{T}$ modulo $\hbar^{1-2\tilde{\delta}}$\ ($k, l \in\Zbb$). 
Thus, the set of accumulation points of the differences among eigenvalues is ${\sigma_{\mbox{\rm e.d.}} (E)} = \frac{2\pi }{T} \Zbb$, as 
desired. 
\end{proof}

\section{Discussions}
Geodesic circle foliations on finite-dimensional manifolds are discussed. 
We then introduce that the periods of Reeb flows and Hamilton flows are exactly the same if the manifold is connected and 
all corresponding flows are periodic. 
Some examples and applications are presented. 

Many other equations can be formulated as such systems. We also expect similar results for infinite-dimensional systems,  
even though our current focus is on finite-dimensional manifolds. We will resume these topics in a separate article.


\appendix 
\def\thesection{Appendix \Alph{section}:}
\makeatother

\section{Geodesible vector fields} \label{sec:counting}
\renewcommand{\theequation}{A.\arabic{equation}}
\setcounter{equation}{0}
 \setcounter{lemma}{0}
    \renewcommand{\thelemma}{\Alph{section}\arabic{lemma}}
Here we heavily borrow the proofs from \cite{Becker, CV, Sullivan}.
To show Proposition \ref{prop: geodesible condition Sullivan}, let us recall the following lemma.

\begin{lemma}\label{lemma: relation between one-form and metric tensor} 
Let $X$ be a vector field on a Riemannian manifold $M$, and let $\alpha := i_X g$. Then
\beq\label{eq: relation between one-form and metric tensor}
d\alpha(X, \cdot)=g (\nabla_X X, \cdot) - \frac{1}{2}d(|X|^2).
\eeq
\end{lemma}
\begin{proof} Let $Y$ be an arbitrary vector field on $M$. Then 
\begin{align*}
d\alpha(X, Y) &= X(\alpha(Y)) -Y(\alpha(X)) -\alpha([X, Y]) \\
&=X(g(X, Y)) - Y(|X|^2) -g(X, [X, Y]) \\
&=g(\nabla_X X, Y) +g(X, \nabla_X Y) -Y(|X|^2) -g(X, [X, Y]) \\
&=g(\nabla_X X, Y) +g(X, \nabla_X Y -[X, Y]) -Y(|X|^2) \\ 
&=g(\nabla_X X, Y) +g(X, \nabla_Y X) -Y(|X|^2) \\ 
&=g(\nabla_X X, Y) - \frac{1}{2} d(|X|^2)(Y). 
\end{align*}
Thus, we have the equation \eqref{eq: relation between one-form and metric tensor}. 
\end{proof}

\begin{proof}[Proof of Proposition \ref{prop: geodesible condition Sullivan}]
Assume that $X$ is geodesible for some Riemannian metric, and let $\alpha := i_X g$. Then it follows from Lemma \ref{lemma: relation between one-form and metric tensor} that $i_X d \alpha =0$ since the right--hand side of the equation \eqref{eq: relation between one-form and metric tensor} vanishes. $\alpha(X) =|X|^2 =1$.  Thus there is a $1$--form $\alpha$ on $M$ such that $\alpha(X)=1$ and $i_X d\alpha=0$.

Conversely, assume that there exists a $1$--form $\alpha$ such that $\alpha(X)=1$ and $\i_X d\alpha=0$. 
Define a metric $g$ by first choosing an arbitrary metric on ${\rm{ker}}\ \alpha$, then setting $|X| \equiv 1$ 
and declaring $X$ and ${\rm{Ker}}\ \alpha$ to be orthogonal. Of course, there are many ways to define
such a metric, but the specific choice is not important for the argument. Given such
a metric $g$, it follows from Lemma \ref{lemma: relation between one-form and metric tensor} that 
$$
0=i_x d\alpha = g (\nabla_X X, \cdot) - \frac{1}{2}d(|X|^2) =g (\nabla_X X, \cdot),   
$$
which implies that $\nabla_X X$ must vanish identically, since g is nondegenerate. Hence, $X$ is geodesible  
with respect to $g$.
\end{proof}

\section{Taut foliations}
We state the condition of taut foliations:  
\begin{lemma}[\cite{Sullivan}]\label{lem: taut condition}
A foliation is taut if and only if there esists a one-form $\tau$ such that $\tau \mbox{(each foliation direction)} > 0$ and 
$d\tau \mbox{(any 2-plane tangent to foliation)} = 0$.
\end{lemma}
\begin{proof}
According to Proposition \ref{prop: geodesible condition Sullivan}, a vector field $X$ is geodesible if and only if $X$ is 
a Reeb vector field. Thus, the corresponding flow is geodesible if and only if there exists a transverse field of codimension-one 
planes that is invariant under the flow. 
The condition stated in the Lemma B2 is a reformulation of this requirement. 
Specifically, an invariant transverse codimension one--plane field and a geodesible parametrization determine a
one--form $\tau$ satisfying $i \cdot \tau = 1$ and $(di + id)\tau = 0$ (or $i d \tau = 0$). 
Conversely, given a form as in the condition in Lemma \label{lem: taut condition} 
choose the parametrization so that $i \cdot \tau = 1$. The second condition
becomes $id\tau = 0$. Thus, $(di + id)\tau = 0$, and the kernel of $\tau$ is the desired invariant. 
This proves Lemma B2 assuming the geodesible condition. 
\end{proof}
It is worth noting that especially if $\theta$ is a nowhere-vanishing closed one-form on a manifold $N$,  
then the $\rm{Ker}\ \theta$ defines a taut foliation.
\end{document}